\documentclass[1p]{elsarticle}
\usepackage{graphicx}
\usepackage{amsmath,amssymb}

\usepackage{hyperref}

\usepackage{setspace}
\usepackage{amsthm}
\newtheorem{lemma}{Lemma}

\theoremstyle{definition}

\begin{document}
\title{Time-headway distribution for periodic totally asymmetric exclusion process with various updates}

\author[fjfi]{P. Hrab\'ak\corref{cor1}}
\ead{pavel.hrabak@fjfi.cvut.cz}
\author[fjfi]{M. Krb\'alek}
\ead{milan.krbalek@fjfi.cvut.cz}
\cortext[cor1]{Corresponding author. Tel.: +420 224 358 567; fax: +420 234 358 643}
\address[fjfi]{Faculty of Nuclear Sciences and Physical Engineering, Czech Technical University in Prague, Trojanova 13, 120 00 Prague, Czech Republic}

\begin{abstract}
The totally asymmetric exclusion process (TASEP) with periodic boundaries is considered as traffic flow model. The large-L approximation of the stationary state is used for the derivation of the time-headway distribution (an important microscopic characteristics of traffic flow)  for the model with generalized update (genTASEP) in both, forward- and backward-sequential representation. The usually used updates, fully-parallel and regular forward- and backward-sequential, are analysed as special cases of the genTASEP. It is shown that only for those cases, the time-headway distribution is determined by the flow regardless to the density. The qualitative comparison of the results with traffic data demonstrates that the genTASEP with backward order and attractive interaction evinces similar properties of time-headway distribution as the real traffic sample.
\end{abstract}

\begin{keyword}
Totally asymmetric exclusion process, time-headway distribution, generalized  ordered-sequential update.
\end{keyword}

\maketitle


\section{Introduction}
Despite its simplicity, the TASEP exhibits some essential aspects of collective motion, and therefore can be considered as one of the simplest models of traffic flow~\cite{ChoSanSch2000PR,SchChoNis2010}. Dynamics of the model is however highly influenced by the used updating procedure. More complicated updates are investigated in order to capture some features of pedestrian dynamics or traffic flow~\cite{RajSanSchSch1998JSP, WolSchSch2006JPA, AppCivHil2011JSM, AppCivHil2011JSMb, DerPogPovPri2012JSM}. The influence of the updates on model dynamics is a motivation for detailed investigation of microscopic traffic-flow characteristics as the time-headway distribution and its dependence on considered updates. The time-headway distribution is a generic characteristics influenced mainly by the density of vehicles or particles. Therefore, the TASEP with periodic boundaries preserving the overall density is considered in this article.

According to our knowledge, the time-headway distribution has been so far investigated only for the TASEP with fully-parallel update in~\cite{ChoPasSin1998EPJB, GhoMajCho1998PRE}, since it corresponds to the Nagel-Schreckenberg model~\cite{NagSch1992JP, Sch1999EPJ}. A preliminary study of the distribution for continuous-time dynamics is presented in~\cite{HraKrb2011Procedia}. In this article, we aim to extend the time-headway studies by the analytical derivation of the distribution for the TASEP with generalized update (genTASEP) introduced in~\cite{DerPogPovPri2012JSM}. The regular ordered updates and the parallel update are considered as special cases of the genTASEP. 

The time-headway distribution for TASEP is investigated in stationary state. The TASEP steady-state solution can be elegantly obtained by means of the Matrix Product Ansatz~\cite{RajSanSchSch1998JSP, DerEvaHakPas1993JoPA, BlyEva2007JoPA}, which has been used for the distance-headway study near the boundaries in~\cite{KrbHra2011JoPA}. The distance-headway distribution for fully-parallel TASEP has been derived by means of car-oriented mean field approximation in~\cite{SchSch1997JoPA} or by means of the (2,1)-cluster approximation~\cite{AvrKoh1992PRA, SchSchNagIto1995PRE}. In this article an alternative method is used for the derivation of the steady-state distribution: the appropriate mapping to the mass transport model, which has been investigated in~\cite{ZiaEvaMaj2004JSM, EvaMajZia2004JoPA} as a generalization of zero-range process~\cite{EvaHan2005JoPA}. The large-$L$ approximation of the stationary state is then used for the derivation of headway distribution.

Without any doubt, the headway distribution is important microscopic characteristics of traffic flow reflecting the repulsive interaction between vehicles. Essential aspects of the headway distribution are described in~\cite{SchChoNis2010, May1990, Treiber2013}. In this article we use traffic data measured on the Dutch highway A9 in 2002 as a reference data sample, which has been used for extensive study of distance-headway distribution in~\cite{Krb2007JoPA} as well. The distribution is typically evaluated with respect to traffic density $\varrho$, therefore we consider  one-parametric families of probability density functions $\{\wp(\,.\,;\varrho)\,;\, \varrho\in[0,\varrho_{c}]\}$ and  $\{f(\,.\,;\varrho)\,;\, \varrho\in[0,\varrho_{c}]\}$, where $\varrho_{c}$ is the density corresponding to the maximal occupancy of the road. Here $\wp(\,.\,;\varrho)$ stands for the distribution of distance-gap $\Delta x$ between the bumpers of two consecutive vehicles, and $f(\,.\,;\varrho)$ is the distribution of time-headway $\Delta t$ between the passages of two consecutive vehicles through given cross-section. To capture the essence of repulsive forces between vehicles it is convenient to scale the headway to mean value equal to one (see~\cite{Krb2007JoPA, KrbHel2004PhysicaA, KrbSebWag2001PRE}). The distribution of the scaled-headway $\Delta t/\langle\Delta t\rangle_\varrho$ will be denoted by $\bar{f}(\,\cdot\,;\varrho)$. The histograms of the empirical probability density functions $f_\mathrm{emp}(t;\varrho)$ and $\bar{f}_\mathrm{emp}(s;\varrho)$ for various intervals of density $\varrho$ [veh/km] are visualized in Figure~\ref{fig:THrealDATA}.

\begin{figure}[h!t]
\centering
	\hfill
	\includegraphics[scale=1]{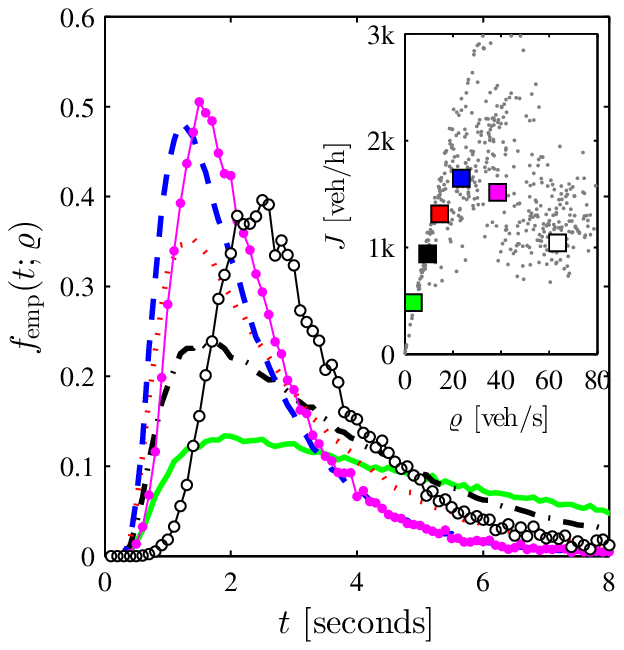}
	\hfill
	\includegraphics[scale=1]{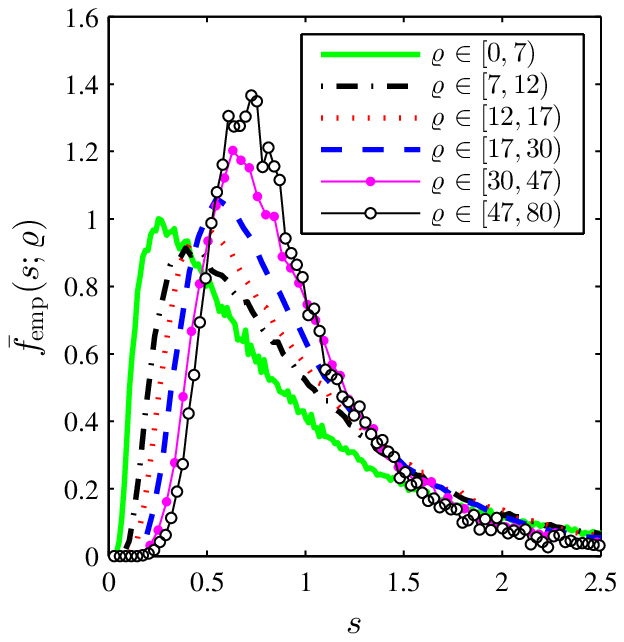}
	\hfill
\caption{Empirical time-headway distribution (left) and scaled-headway distribution (right) for different density regions $\varrho$ [veh/km]. The sub-plot visualizes average flow $J$ corresponding to given density region (squares) with empirical flow-density diagram in the background. Visualized data are extracted from the measurement on the Dutch freeway A9 in 2002.}
\label{fig:THrealDATA}
\end{figure}

Analyzing headway distributions measured on Dutch freeway we can (in agreement with~\cite[Section 2.1]{May1990} and~\cite[Section 6.8]{SchChoNis2010}) observe following aspects:
\begin{itemize}
	\item The time-headway modus is always less than the median, which is less then the mean. With increasing flow, these statistics converge towards each other.
	\item With increasing flow, the mean value and the variance of the time-headway is decreasing.
	\item The modus of the distribution corresponding to densities higher then maximal-flow density is shifted towards higher values in comparison with the distribution corresponding to lower density and similar flow.
	\item The variance of the scaled-headway distribution is decreasing with respect to $\varrho$ (similarly as in the case of distance-headways~\cite{Krb2007JoPA} and multiheadways~\cite{Krb2013JoPA}).
\end{itemize}

\section{TASEP and Considered Updates}

Let $\mathbb{L}=\left\{0, 1, \dots, L-1\right\}$ be a linear lattice consisting of $L$ equivalent sites. The periodicity of the lattice is expressed in the sense that for each $x,y\in\mathbb{L}$ it holds $x+y:=(x+y \bmod L)$. There are $N$ indistinguishable particles moving along the lattice $\mathbb{L}$ by hopping from one site to the neighbouring site. Particles follow the \emph{exclusion rule}, i.e, every site can be occupied by at most one particle. In the following we understand by the density $\varrho$ the average occupation of a site given by $\varrho=N/L$. As the density $\varrho\in[0,1]$ will often be used as the system parameter, it is implicitly supposed that the number of particles is given by $N=\lfloor\varrho L\rfloor$.

In TASEP a particle sitting in the site $x$ hops to the neighbouring site $x+1$ with probability $p\in(0,1]$, if the target site is empty. Those hops are driven by updates listed below. In this article we consider fully-parallel, forward-sequential, and backward-sequential update according to the nomenclature in \cite{RajSanSchSch1998JSP} as special cases of the generalized update~\cite{DerPogPovPri2012JSM}.


{\em Fully-parallel update}: All sites are updated simultaneously during one algorithm step. If convenient, this update will be indicated by the parallel sign $\parallel$. In the case of {fully-parallel} update the jump of a particle can be considered as an exchange of the particle and the neighbouring empty site (hole). Therefore, the motion of particles in the system with density $\varrho$ corresponds to the motion of holes with density $\sigma=1-\varrho$ in opposite direction. This is referred to as the {\em particle-hole symmetry}. Similarly evinces the particle-hole symmetry the random-sequential and continuous-time update.

{\em Ordered-sequential update}: The sites are updated sequentially according to given order, i.e., we consider site-oriented definition of the updates. We focus on the backward and forward update, i.e., sites are updated in the decreasing or increasing order respectively. If convenient, these updates will be indicated by the left ($\leftarrow$) or right ($\rightarrow$) arrows sign respectively. The backward-sequential update can (and should) be understand in the way that exactly $k$ particles from a block of $n$ particles will move by one site forward with probabilities
\begin{equation}
\label{eq:pkn}
	p(k|n)=\left\{\begin{array}{ll}
		(1-p)p^k & k< n\,,\\
		p^n & k=n\,.
	\end{array}\right.
\end{equation}
Analogically, in the forward-sequential case a particle having exactly $n$ empty sites in front of it hops $k$ sites forward with probability $p(k|n)$ defined by~\eqref{eq:pkn}.

{\em Generalized update}: This update has been introduced in~\cite{DerPogPovPri2012JSM} and further investigated in~\cite{DerPovPri2015PRE}. Originally, the sites are updated in backward order. The first particle in a block hops to the neighbouring site with probability $p$ as usual, but further particles in the block hop with a modified probability $p\gamma$, where $\gamma\in[0,1/p]$ is a parameter influencing the repulsion ($\gamma<1$) or attraction ($\gamma>1$) of particles. Here we note that in the original article, the parameter of the generalized order is $\delta=\gamma-1\in[-1,1/p-1]$.

This update can again be understood in the way that exactly $k$ particles from a block of $n$ particles will move by one site forward with probabilities
\begin{equation}
\label{eq:gpkn}
	p(k|n)=\left\{\begin{array}{ll}
		1-p & k=0\,,\\
		p(p\gamma)^{k-1}(1-p\gamma) & 0<k< n\,,\\
		p(p\gamma)^{n-1} & k=n\,.
	\end{array}\right.
\end{equation}
Similarly, we can define a forward variant of such update enabling a particle to hop $k$ sites forward with probability~\eqref{eq:gpkn}. 

It is worth noting that the forward- and backward-sequential updates evinces a dualism in motion of particles and holes, i.e., the motion of particles with density $\varrho$ in forward update corresponds to the backward-update-driven motion of holes with density $1-\varrho$ in opposite direction and vice versa.

From the dynamics it is obvious that the fully parallel update is a special case of the generalized update with $\gamma=0$. Similarly, the backward-sequential update is a generalized update with $\gamma=1$. From these reasons it is sufficient to consider only the generalized update.

For readability reasons we use the notation $q=1-p$ and $\sigma=1-\varrho$.

\section{Stationary State}

The TASEP driven by above mentioned updates can be considered as Markov process (for more detail about Markov processes see e.g.~\cite{Bremaud}) with finite state space $X_L=\{0,1\}^L$. Every state $\boldsymbol{\tau}\in X$ is a vector $\boldsymbol{\tau} =(\tau_0,\dots,\tau_{L-1})$, where $\tau_x=0$ if the site is empty and $\tau_x=1$ if the site is occupied. Let the process $\big(\boldsymbol{\tau}(t), t\in T\big)$  be defined on the canonical trajectory probability space $\left((X_L)^T, \mathcal{B}((X_L)^T), \Pr\right)$, where $T=\{k\mid k\in\mathbb{Z}\}$ represents discrete-time set , $\mathcal{B}$ denotes the Borel sigma field, and $\Pr$ is the canonical probability measure. Let us denote by 
\begin{equation}
	W(\boldsymbol{\tau}\to\boldsymbol{\tau}')=\Pr\left[\boldsymbol{\tau}(t+1)=\boldsymbol{\tau}'\mid\boldsymbol{\tau}(t)=\boldsymbol{\tau}\right]
\end{equation}
the transition rates from state $\boldsymbol{\tau}$ to state $\boldsymbol{\tau}'$. The particle preserving dynamics and the periodic lattice imply that the number of particles $N$ is constant for each trajectory, therefore the state space $X_L$ decomposes into $N+1$ irreducible subspaces
\begin{equation}
	X_{L,N}=\left\{\boldsymbol{\tau}\in X \mid \textstyle{\sum_{x=0}^{L-1}}\,\tau_x=N \right\}\,,\quad N=0,1,\dots,L\,.
\end{equation}

On each subspace $X_{L,N}$ we define the \emph{stationary measure} $\mathcal{P}_{L,N}$ as the solution of the equation set
\begin{equation}
    \sum_{\boldsymbol\tau'\neq\boldsymbol\tau}
    W(\boldsymbol\tau\to\boldsymbol\tau')\mathcal{P}(\boldsymbol\tau)=
    \sum_{\boldsymbol\tau'\neq\boldsymbol\tau}
    W(\boldsymbol\tau'\to\boldsymbol\tau)\mathcal{P}(\boldsymbol\tau')\,,\quad \boldsymbol{\tau}\in X_{L,N}\,.
\end{equation}

Let us further define the \emph{cluster probability} measure $\mathcal{P}_{L,N}^n$ for $n\in\mathbb{N}$, $n\leq L$, as
	\begin{equation}
		\mathcal{P}_{L,N}^n(s_1,s_2,\dots,s_n)=\mathcal{P}_{L,N}\left(\tau_{x+j}=s_j\,,\,j=1,2,\dots,n\right)\,,
	\end{equation}
where $x$ is an arbitrary site from $\mathbb{L}$, and $\mathcal{P}_{L,N}$ is the stationary measure on $X_{L,N}$. Note that the cluster probability is well defined since the periodic lattice consists of equivalent sites and therefore the measure $\mathcal{P}_{L,N}$ is translation invariant. Furthermore, in relation to the density $\varrho$ it holds $\mathcal{P}_{L,\lfloor\varrho L\rfloor}^{1}(1)=\varrho$, $\mathcal{P}_{L,\lfloor\varrho L\rfloor}^{1}(0)=\sigma$.

As common we will use the large $L$ approximation of the stationary measure. By the large $L$ approximation we understand the approximation of cluster probabilities by their limits
\begin{equation}
	\mathcal{P}_{L,\lfloor\varrho L\rfloor}^{n}(s_1,\dots,s_n)\approx\lim_{L\to+\infty}\mathcal{P}_{L,\lfloor\varrho L\rfloor}^{n}(s_1,\dots,s_n)=:\mathcal{P}_{\varrho}^{n}(s_1,\dots,s_n)\,,
\end{equation}
for all $n\in\mathbb{N}$, if the limits exist. The large $L$ approximation measures $\mathcal{P}_{\varrho}^n$ can be considered as measures of the process defined on the infinite line. For details concerning the convergence of the large $L$ approximation to the infinite line measure see book~\cite{Liggett}, article~\cite{GroSchSpo2003JSP} and references therein.

In further text we will use following notation for the conditional probabilities:
\begin{eqnarray}
	\mathcal{P}(\underline{s_1}s_2\dots s_n\mid s_1)&:=&\mathcal{P}(\tau_{x+j}=s_j, j=1,\dots,n\mid \tau_{x+1}=s_1)=\nonumber\\
	&=&\mathcal{P}^{n}(s_1,\dots,s_n)/\mathcal{P}^{1}(s_1)\,.
\end{eqnarray}

In this notation we can express the probability that there are exactly $d$ empty cells between two consecutive particles as 
\begin{equation}
		\wp(d;\varrho)=\mathcal{P}_{\varrho}(\underline{1}\overbrace{0\dots0}^d 1\mid 1)=\mathcal{P}_{\varrho}^n(1\overbrace{0\dots0}^d 1)/\varrho\,,\quad d\geq0\,.
\end{equation}
The distribution $\wp(\cdot;\varrho)$ is referred to as the distance-headway distribution. Let us in the following consider the lattice to be sufficiently large and let us use the large $L$ approximations $\mathcal{P}_{\varrho}$ for $\varrho\in(0,1)$. In such case it holds
\begin{equation}
 \sum_{d=0}^{+\infty}\wp(d;\varrho)=1\,,\quad\sum_{d=0}^{+\infty}(d+1)\wp(d;\varrho)=\frac{1}{\varrho}\,.
\end{equation}

To obtain the stationary measure $\mathcal{P}_\varrho$, it is convenient to map the generalized TASEP on the mass transport process~\cite{ZiaEvaMaj2004JSM,EvaMajZia2004JoPA}, for which the stationary distribution is known. The derivation can be found in the~\ref{app:ssd}. For the backward sequential update then follows that
\begin{eqnarray}
	\mathcal{P}_{\leftarrow,\varrho}(\underline{0}\overbrace{1,\dots,1}^{n}0\mid 0)&=&\phi(n;\varrho)=
	\left\{\begin{array}{ll}
		1-z/\sigma & n=0\,,\\
		z^2/(\varrho\sigma)(1-z/\varrho)^{n-1} & n\geq1\,,
	\end{array}\right.\\
	\wp_\leftarrow(d;\varrho)&=&\phi(n;\sigma)=
	\left\{\begin{array}{ll}
		1-z/\varrho & n=0\,,\\
		z^2/(\varrho\sigma)(1-z/\sigma)^{n-1} & n\geq1\,,
	\end{array}\right.
\end{eqnarray}
where
\begin{equation}
	z=z(\varrho,\gamma,p)=\frac{1-\sqrt{1-4\varrho\sigma A(p,\gamma)}}{2A(p,\gamma)}\,,\quad A(p,\gamma)=\frac{p(1-\gamma)}{1-p\gamma}\,,
\end{equation}
is one of the roots of the quadratic equation
\begin{equation}
\label{eq:z}
	A(p,\gamma)z^2-z+\varrho\sigma=0\,.
\end{equation}

For the flow $J$ then holds
\begin{equation}
\label{eq:J}
	J_{\leftarrow}(\varrho;p,\gamma)=\sum_{n=1}^{+\infty}\mathcal{P}_{\leftarrow,\varrho}^{n+1}(\overbrace{1\dots1}^{n}0)p(p\gamma)^{n-1}=\frac{pz}{1-p\gamma(1-z/\varrho)}\,.
\end{equation}

\section{Time-Headway Distribution for Parallel Updates}

This section extends the derivation of the time-headway distribution for fully-parallel update~\cite{ChoPasSin1998EPJB, GhoMajCho1998PRE} by the derivation of the time-headway distribution for generalized update. In the following we will use the duality of forward and backward update procedure. For technical reasons, the derivation is performed for the forward variant of the update. The transformation to the backward variant is straightforward by switching the symbols $\varrho$ and $\sigma$.

Let us denote by $k_\mathrm{in}^{\alpha}(x)$ the time (in steps of the algorithm) at which the particle $\alpha$ enters the site $x$ and by $k_\mathrm{out}^{\alpha}(x)$ the time at which the particle $\alpha$ hops out of the site $x$. Let site $x=0$ be the reference site for time-headway measurement. Let us further consider two consecutive particles, which will be referred to as the \emph{leading particle} (abbr. LP) and the \emph{following particle} (abbr. FP). The aim of following calculations is to derive the distribution of the time interval $\Delta k=k_\mathrm{out}^\mathrm{FP}(0)-k_\mathrm{out}^\mathrm{LP}(0)$.

The \emph{time-headway distribution} is defined by means of the \emph{time-headway probabilities} $f(k)$, $k\in\mathbb{N}$, where $f(k;\varrho)=\Pr(\Delta k=k\mid\varrho)$. The own derivation of the time-headway distribution is provided in~\ref{app:thd}. The final formula for the forward variant of the genTASEP model for $p\in(0,1)$, $\gamma\in[0,1/p)$, and $\varrho\in(0,1)$ is
\begin{eqnarray}
\label{eq:fkfinfor}
	f_{\to}(k;\varrho,p,\gamma)
	&=& \frac{pz}{\sigma-z}\frac{1}{1-\omega}\left(1-\frac{pz}{\sigma}\frac{1}{1-\omega}\right)^{k-1}\nonumber\\
	&+&\frac{pz}{\varrho-z}(1-\omega)\left(1-\frac{pz}{\varrho}\right)^{k-1}\nonumber\\
	&-&\left[\frac{pz}{\sigma-z}(1+\omega)+\frac{pz}{\varrho-z}(1-\omega)\right]q^{k-1}\nonumber\\
	&-&p^2\frac{1-p\gamma}{1-p}q^{k-1}(k-1)\,,
\end{eqnarray}
where $\omega=\omega(\varrho,p,\gamma)=1-p\gamma(1-z/\sigma)$ and
\begin{equation}
	z=z(\varrho,\gamma,p)=\left(1-\sqrt{1-4\varrho\sigma\frac{p(1-\gamma)}{1-p\gamma}}\right)\left(\frac{2p(1-\gamma)}{1-p\gamma}\right)^{-1}\,.
\end{equation}

The time-headway distribution for backward variant of the genTASEP is then
\begin{eqnarray}
\label{eq:fkfinbac}
	f_{\leftarrow}(k;\varrho,p,\gamma)
	&=& \frac{pz}{\varrho-z}\frac{1}{1-\bar{\omega}}\left(1-\frac{pz}{\varrho}\frac{1}{1-\bar{\omega}}\right)^{k-1}\nonumber\\
	&+&\frac{pz}{\sigma-z}(1-\bar{\omega})\left(1-\frac{pz}{\sigma}\right)^{k-1}\nonumber\\
	&-&\left[\frac{pz}{\varrho-z}(1+\bar{\omega})+\frac{pz}{\sigma-z}(1-\bar{\omega})\right]q^{k-1}\nonumber\\
	&-&p^2\frac{1-p\gamma}{1-p}q^{k-1}(k-1)\,,
\end{eqnarray}
where $\bar{\omega}=\bar{\omega}(\varrho,p,\gamma)=1-p\gamma(1-z/\varrho)$.

Putting $\gamma=0$ we obtain the known formula for parallel update
\begin{eqnarray}
\label{eq:fkpar}
	f_{\parallel}(k;\varrho,p)&=&f_{\to}(k;\varrho,p,0)=f_{\leftarrow}(k;\varrho,p,0)=\nonumber\\
	&=& \frac{py}{\sigma-z}\left(1-\frac{py}{\sigma}\right)^{k-1}  + \frac{py}{\varrho-y}\left(1-\frac{py}{\varrho}\right)^{k-1}\nonumber\\
	&-&\left[\frac{py}{\sigma-z}+\frac{py}{\varrho-z}\right]q^{k-1} - p^2 q^{k-2}(k-1)\,,
\end{eqnarray}
where $y=y(\varrho,p)=z(\varrho,p,0)$.

Similarly we can obtain the distribution for forward-sequential update with $\gamma=1$ resulting to
\begin{equation}
\label{eq:forfk}
    f_{\rightarrow}(k;\varrho,p,1)=
    \frac{p\varrho}{q\sigma}\left(\frac{q}{1-p\sigma}\right)^k + \frac{p\sigma}{\varrho}(1-p\sigma)^k -\left[\frac{p\varrho}{q\sigma}+\frac{p\sigma}{\varrho}\right]q^k -p^2kq^{k-1}
\end{equation}

The convergence of the time-headway distribution corresponding to finite lattice size $L<<\infty$ to the large $L$ approximation~\eqref{eq:forfk} is demonstrated in Figure~\ref{fig:THforSIML}. The histogram estimation $\hat{f}_{\to}^\mathrm{sim}$ of the probability density function $f_{\to}$ based on the computer  simulation of the model with finite lattice size $L$ is displayed. We can observe that the formula~\eqref{eq:forfk} fits sufficiently the simulation data for $L=50$.

\begin{figure}
\centering
	\includegraphics[scale=1]{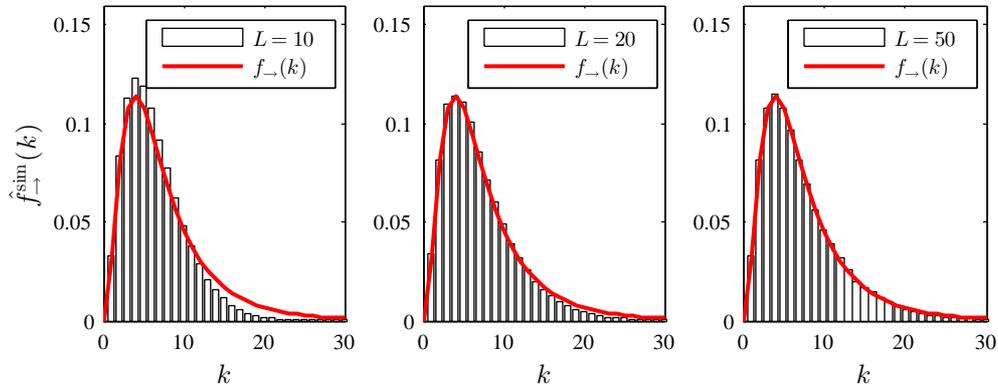}
\caption{Comparison of empirical probability density function $\hat{f}_{\to}^\mathrm{sim}$with the derived formula~\eqref{eq:forfk}. Number of observations $n=2\cdot 10^5$, $\varrho=0.2$, $\gamma=1$, lattice length $L=10$ (left), $L=20$ (center), $L=50$ (right).}
\label{fig:THforSIML}
\end{figure}

Let us for completeness present the time-headway distribution $f_{ct}(t;\varrho)$ for continuous time dynamics, which can be obtained from the fully-parallel update distribution by the limit $p\to0_+$ with the appropriate time scaling $t=pk$. The distribution function is
\begin{equation}
	F_{ct}(t;\varrho)=\lim_{p\to0_+}\sum_{k\leq t/p}f_{\parallel}(k;\varrho,p)\,,
\end{equation}
which leads to probability density function
\begin{equation}
\label{eq:ft}
	f_{ct}(t;\varrho)=\frac{\varrho}{\sigma}\left(\mathrm{e}^{-\varrho t}-\mathrm{e}^{-t}\right)+\frac{\sigma}{\varrho}\left(\mathrm{e}^{-\sigma t}-\mathrm{e}^{-t}\right)-t\mathrm{e}^{-t}\,,\quad t\geq0\,.
\end{equation}

\section{Mutual relations and correspondence}

A typical characteristics of the fully-parallel update is the particle-hole symmetry mentioned above. In processes with particle-hole symmetry the holes behave exactly like particles in system with complementary density $\sigma=1-\varrho$. This projects to the stationary distribution in the way that $\forall n\in\mathbb{N}$ and $\forall(s_1,\dots,s_n)\in\{0,1\}^n$ it holds
\begin{equation}
\label{eq:phs}
	\mathcal{P}_{\varrho}^{n}(s_1,\dots,s_n)=\mathcal{P}_{\sigma}^{n}(1-s_1,\dots,1-s_n)\,.
\end{equation}
Furthermore, the symmetry in the hopping probabilities implies that the flow and time-headway distribution are symmetrical with respect to $\varrho$ as well, i.e.,
\begin{equation}
	J(\varrho)=J(\sigma)\,,\qquad f(\cdot;\varrho)=f(\cdot;\sigma)\,.
\end{equation}

There is a dualism between forward and backward dynamics (both regular and generalized) similar to the particle-hole symmetry. The motion of holes in forward update corresponds to the motion of particles in the backward update and vice versa. Therefore it holds
\begin{equation}
\label{eq:fbdual}
	\mathcal{P}_{\rightarrow;\varrho}^{n}(s_1,\dots,s_n)=\mathcal{P}_{\leftarrow;\sigma}^{n}(1-s_1,\dots,1-s_n)\,, \quad J_{\rightarrow}(\varrho)=J_{\leftarrow}(\sigma)\,,\quad f_{\rightarrow}(\cdot;\varrho)=f_{\leftarrow}(\cdot;\sigma)\,.
\end{equation}
The particle-hole symmetry for fully-parallel update can be illustrated by the fact that
\begin{equation}
	f_{\parallel}(k;\varrho,p)=f_{\rightarrow}(k;\varrho,p,0)=f_{\leftarrow}(k;\varrho,p,0)\,.
\end{equation}

In this section we discuss under which conditions there can be found similar relation to~\eqref{eq:phs} for $\gamma\neq0$, namely for which $\gamma$ it holds that $\forall\varrho_1,\varrho_2(0<\varrho_1<\varrho_{\mathrm{max}}<\varrho_2<1)$
\begin{equation}
\label{eq:Jtof}
	J(\varrho_1)=J(\varrho_2) \implies f(\cdot;\varrho_1)=f(\cdot;\varrho_2)\,.
\end{equation}
Here we note that the opposite implication is always fulfilled since $J(\varrho)=(\mathbb{E}_{f(\cdot;\varrho)}\Delta t)^{-1}$, where $\mathbb{E}_f$ stands for the expectation value with respect to distribution $f$. The property~\eqref{eq:Jtof} means that the time-headway distribution is fully determined by the flow regardless to the density.

The flow $J$ defined by the equation~\eqref{eq:J} is a continuous function with respect to density $\varrho$ with one maximum $J_{\mathrm{max}}$ at $\varrho_{\mathrm{max}}$ and with $J(0)=0=J(1)$. Therefore $\forall \varrho_1\in(0,\varrho_{\mathrm{max}})$ there exists unique $\varrho_2\in(\varrho_{\mathrm{max}},1)$ for which $J(\varrho_1)=J(\varrho_2)$. E.g. in parallel dynamics $\varrho_2=1-\varrho_1$. Figure~\ref{fig:fkgamma} shows eight sub-plots of the probability density $f_{\leftarrow}(\cdot;\varrho,p,\gamma)$ for different values of $\gamma$. Each sub-plot contains the probability density for $\varrho=\varrho_{\mathrm{max}}$ and $\varrho_1, \varrho_2$ corresponding to the flow $J_{\mathrm{max}}/2$.

\begin{figure}[htb]
\centering
	\includegraphics[width=\textwidth]{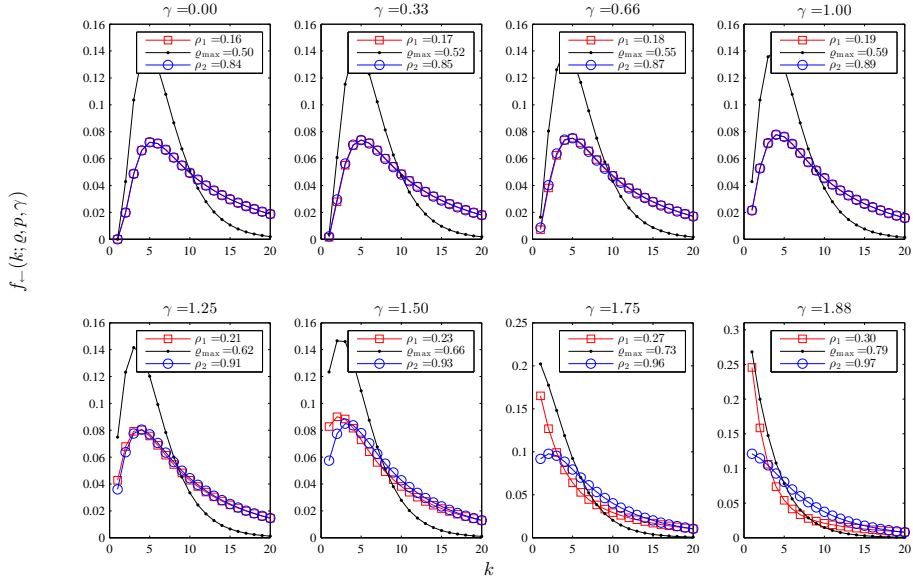}
\caption{The time-headway distribution $f_{\leftarrow}(\cdot,\varrho,p,\gamma)$ for $p=0.5$ and $\gamma\in\{0, 1/3, 2/3, 1, 5/4, 3/2, 7/4, 2-\epsilon\}$. Black dots ($\bullet$) correspond to the maximal flow density $\rho_\mathrm{max}$; $\rho_1<\rho_\mathrm{max}$ ($\square$) and $\rho_2>\rho_\mathrm{max}$ ($\bigcirc$) correspond to the flow $J_\mathrm{max}/2$.}
\label{fig:fkgamma}
\end{figure}

Figure~\ref{fig:df1f2} shows the symmetrical $\chi^2$ distance $d(\cdot,\cdot)$ between $f_{\leftarrow}(k;\varrho_1,p,\gamma)$ and $f_{\leftarrow}(k;\varrho_2,p,\gamma)$ for $p=0.5$ with respect to $\gamma\in[0,2]$, where
\begin{equation}
\label{eq:df1f2}
	d(f_1,f_2)=\sum_{k=1}^{+\infty}\frac{[f_1(k)-f_2(k)]^2}{f_1(k)+f_2(k)}\,.
\end{equation}

\begin{figure}[htb]
\centering
	\includegraphics[width=.8\textwidth]{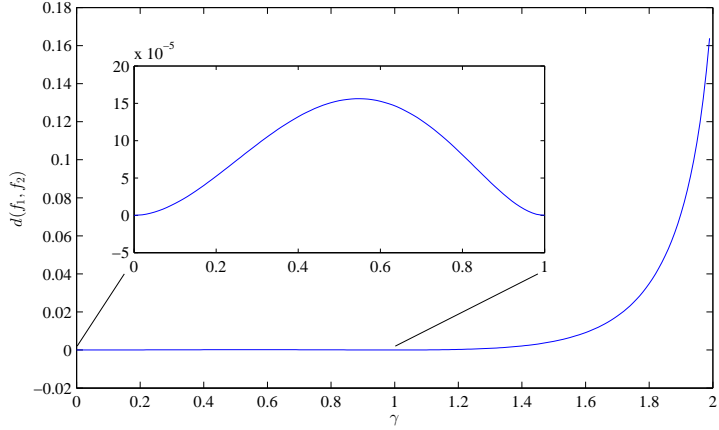}
\caption{The visualisation for of the distance~\eqref{eq:df1f2} between $f_1=f_{\leftarrow}(\cdot,\varrho_1,p,\gamma)$ and $f_2=f_{\leftarrow}(\cdot,\varrho_2,p,\gamma)$ for $p=0.5$ and $J(\varrho_1)=J(\varrho_2)=J_\mathrm{max}/2$.}
\label{fig:df1f2}
\end{figure}

From these graphs we may conclude that the implication~\eqref{eq:Jtof} holds only for the fully parallel update $\gamma=0$ and regular backward- (or forward-) sequential update $\gamma=1$. Nevertheless, the difference of the distributions for repulsive dynamics, i.e., $\gamma\in(0,1)$, is very small. For the attractive dynamics with $\gamma>1$, the time-headway distribution is not even closely determined by the flow $J$, which is an interesting property in comparison to the real traffic data.

Let us now focus on the cases, for which the implication~\eqref{eq:Jtof} holds. In the case of fully-parallel update $\gamma=0$ the implication holds due to the particle hole symmetry. To prove the case of regular froward- or backward-sequential update we show a relation between the fully-parallel and ordered-sequential updates. We claim that for $p\in(0,1)$ there is a one-to-one mapping $b:(0,1)\to(0,1)$ such that for each $k=\{1,2,\dots\}$ it holds
\begin{equation}
\label{eq:fkforbacpar}
	f_{\leftarrow}(k;1-\varrho)=f_{\rightarrow}(k;\varrho)=f_{\parallel}\left(k+1;b^{-1}(\varrho)\right)\,.
\end{equation}
The first equality is fulfilled due to the dualism between forward and backward update. To prove the second equality it is sufficient to compare the formulas~\eqref{eq:fkpar} and~\eqref{eq:forfk}.
Then we obtain that
\begin{equation}
	b(\varrho)=1-\frac{y}{\varrho}\,,\quad y=\left(1-\sqrt{1-4\varrho\sigma}\right)/2\,.
\end{equation}
Then for each $\varrho_1<\varrho_{\rightarrow}^{\mathrm{max}}$ there exists unique $\varrho_{\parallel}=b^{-1}(\varrho_1)<\varrho_{\parallel}^{\mathrm{max}}=1/2$. The corresponding $\varrho_2>\varrho_{\rightarrow}^{\mathrm{max}}$ can be then found as $b(1-\varrho_{\parallel})$. All together we get
\begin{equation}
	\varrho_2=b\left(1-b^{-1}(\varrho_1)\right)\,.
\end{equation}
This proves that in case of the sequential update the time-headway distribution is determined by the flow regardless to the density.

Furthermore, the relation~\eqref{eq:fkforbacpar} implies that the regular backward- and forward-sequential update changes the behaviour of the fully parallel TASEP only by the scaling of the density via the bijection $b: \varrho_{\rightarrow}\mapsto\varrho_{\parallel}$. Due to~\eqref{eq:fkforbacpar} it is possible to relate the exact value of the flow $J$ for these updates with corresponding densities, since
\begin{equation}
	\mathbb{E}_{f_{\rightarrow}(\cdot;b(\varrho))}\Delta k=\sum_{k=1}^{+\infty}kf_{\rightarrow}(k;b(\varrho))=\sum_{k=1}^{+\infty}kf_{\parallel}(k+1;\varrho)=\mathbb{E}_{f_{\parallel}(\cdot;\varrho)}\Delta k-1
\end{equation}
and therefore
\begin{equation}
	J_{\leftarrow}(1-b(\varrho))=J_{\rightarrow}(b(\varrho))=\frac{J_{\parallel}(\varrho)}{1-J_{\parallel}(\varrho)}\,.
\end{equation}
This is analogical result as in~\cite[Section 5]{EvaRajSpe1999JSP} for open boundaries.

\section{Conclusions}

This paper provides analytical derivation of the time-headway distribution for TASEP with generalized update (genTASEP) in both, backward- and forward-sequential representation. We provide formulas~\eqref{eq:fkfinfor} and~\eqref{eq:fkfinbac} for the probability density function $f(k;\varrho,p,\gamma)$ for $p\in(0,1)$, $\varrho\in(0,1)$, and $\gamma\in[0,1/p)$. The limiting cases, i.e., $p=1$ and $\gamma=1/p$ can be obtained using the appropriate limit.

Comparing the results for the fully-parallel ($\gamma=0$) and regular ordered-sequential update ($\gamma=1$), we have found out that there is a density-scaling bijection $b: \varrho_{\rightarrow}\mapsto\varrho_{\parallel}$, which relates the fully-parallel and regular ordered-sequential update via the relation~\eqref{eq:fkforbacpar}. The reason for this correspondence lies in the fact that the hopping probability in sequential update does not depend on the position of the particle in the block, therefore the dynamics is the same even if we look at the state in the middle of the update. A consequence of this correspondence is the fact that in the cases $\gamma\in\{0,1\}$ the time-headway distribution is determined by the flow regardless, whether the system density is above or below the maximal flow density $\varrho_{\mathrm{max}}$.

This determination breaks down for other values of the parameter $\gamma\notin\{0,1\}$. Nevertheless, the Figure~\ref{fig:df1f2} demonstrates that for repulsive values of $\gamma\in(0,1)$, the difference between the corresponding distributions is very low. Such a difference significantly increases with the attractive values of $\gamma\in(1,1/p)$. We can observe qualitatively similar behaviour to the traffic sample for the generalized backward-sequential update, originally introduced in~\cite{DerPogPovPri2012JSM}.

\begin{figure}[htb]
\centering
	\includegraphics[width=.5\textwidth]{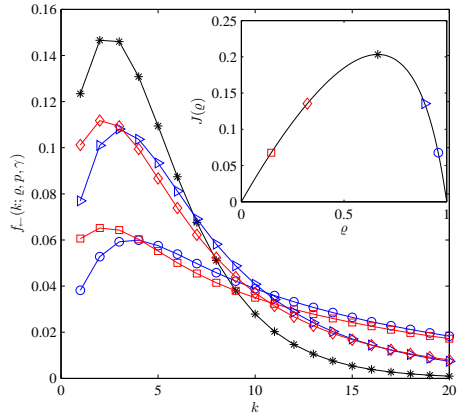}
\caption{Illustration of the distribution $f_{\leftarrow(\cdot,\varrho,p,\gamma)}$ with $p=0.5$, $\gamma=1.5$. The densities are chosen to correspond with $J_{\mathrm{max}}/3$ and $J_{\mathrm{max}}\cdot 2/3$, as depicted in the fundamental diagram in the sub-figure.}
\label{fig:fkrho}
\end{figure}

Figure~\ref{fig:fkrho} shows two pairs of the time-headway distributions $f_{\leftarrow}(\cdot;\varrho_1)$, $f_{\leftarrow}(\cdot;\varrho_2)$ corresponding to the same flow. In general, we can see that the distribution corresponding to the density $\varrho_2>\varrho_{\mathrm{max}}$ (blue line) has lower variance and modus shifted more to the higher values of $k$ than the corresponding distribution for $\varrho_1<\varrho_{\mathrm{max}}$ (red line). Such characteristics has been observed in the real traffic data (compare with the Figure~\ref{fig:THrealDATA} and list of properties below it).

This means that that the synchronized motion of particles in genTASEP mirrors in the time-headway distribution similarly to the real traffic. Therefore, a model based on the genTASEP can describe well the traffic sample even on the microscopic level characterized by the time-headway distribution. The derived formulas~\eqref{eq:fkfinfor} and~\eqref{eq:fkfinbac} can then serve as an additional calibration tool for finding appropriate values of $p$ and $\gamma$. To develop a comprehensive technique of the parameter estimation is a subject of further research.

\section*{Acknowledgement}
This work was supported by the Czech Science foundation under the grant GA\v{C}R 15-15049S and the CTU grant SGS15/214/OHK4/3T/14. 

\appendix

\section{Steady state derivation}
\label{app:ssd}

For the steady state derivation we use the result of \cite{ZiaEvaMaj2004JSM} by mapping the TASEP with generalized update to appropriate mass transport process (abbr. MTP).

Consider a mass transport process of $\tilde{N}$ particles on a lattice of $\tilde{L}$ sites with periodic boundaries and parallel dynamics. In every time step from each site $j$ containing $n_{j}$ particles hops $k_j\leq n_j$ particles to the neighbouring site with probability $\varphi(k_j|n_j)$. The stationary distribution $\Phi(\bf n)$ of such process can be written in a factorized form
\begin{equation}
	\Phi(n_1,\dots,n_{\tilde{L}})=Z({\tilde{L},\tilde{N}})^{-1}\prod_{j=1}^{\tilde{L}}\phi({n_{j}})
\end{equation}
\emph{if and only if} the cross ratio
\begin{equation}
	R(k,n)=\frac{\varphi(k+1|n+2)\varphi(k|n)}{\varphi(k+1|n+1)\varphi(k|n+1)}
\end{equation}
for $0\leq k\leq n$ with $\varphi(0|0)=1$ does not depend on $k$, i.e., $R(n):=R(k,n)$. The one-site factors $\phi(n)$ are then in the form
\begin{equation}
	\phi(n)=\phi(0)\left(\frac{\phi(1)}{\phi(0)}\right)^{n}\cdot\prod_{j=0}^{n-2}R(j)^{-n+j+1}\,\quad n\geq2\,.
\end{equation} 
Let us assume for simplicity that the one site factors $\phi(n)$ fulfil
\begin{equation}
\label{eq:Aconst}
	\sum_{n=0}^{\infty}\phi(n)=1\quad {\rm and} \quad\sum_{n=0}^{\infty}n\phi(n)=\tilde{\varrho}:=\frac{\tilde{N}}{\tilde{L}}\,.
\end{equation}
This means that $\phi(n)$ are one site marginals of the product measure in the large $\tilde{L}$ approximation.

The TASEP with generalized update can be mapped to the MTP with chipping function
\begin{equation}
	\varphi(k|n)=\begin{cases}
		1-p& k=0\,,\\
		p(p\gamma)^{k-1}(1-p\gamma)& 0<k<n\,,\\
		p(p\gamma)^{n-1}& k=n\,,
	\end{cases}
\end{equation}
for $n\geq1$ and $\varphi(0|0)=1$. Hence the factors $\phi(n)$ are
\begin{equation}
	\phi(n)=\phi(0)\left(\frac{\phi(1)}{\phi(0)}\right)^{n}\left(\frac{1-p}{1-p\gamma}\right)^{n-1}\,\quad n\geq2\,.
\end{equation}
Using \eqref{eq:Aconst} and little computational effort we obtain
\begin{equation}
	{\phi(0)=1-\frac{z}{\sigma}\,,\qquad
	\phi(n)=\frac{z^{2}}{\varrho\sigma}\left(1-\frac{z}{\varrho}\right)^{n-1}\,,\quad n\geq1\,,}
\end{equation}
where
\begin{equation}
	z=z(\rho,\gamma,p)=\frac{1-\sqrt{1-4\rho\sigma\frac{p(1-\gamma)}{1-p\gamma}}}{2\frac{p(1-\gamma)}{1-p\gamma}}
\end{equation}

The mass transport process can be transformed to a process following the exclusion rule, if we associate the piles of the MTP with the clusters of particles, i.e., the pile of $n$ particles corresponds to a cluster of $n$ particles followed by an empty site as suggested e.g. in~\cite{KauMahHarPRE2005}. Then we have $\tilde{L}=L-N$, $\tilde{N}=N$ and
\begin{equation}
\label{eq:Arho}
	\tilde{\varrho}=\frac{\tilde{N}}{\tilde{L}}=\frac{N}{L-N}=\frac{\varrho}{1-\varrho}\,.
\end{equation}

The stationary measure $\mathcal{P}_{\varrho}$ is then related to the one-site factors via
\begin{equation}
\label{eq:Aphin}
	\phi(n)=\mathcal{P}_{\varrho}(0\overbrace{1\dots{1}}^{n}\underline{0}\mid \tau_{x}=0)\,.
\end{equation}
To express any probability $\mathcal{P}_{\varrho}^n$ in the terms of $\phi(n)$ it is sufficient to use repeatedly the \emph{Kolmogorov consitency conditions} (mentioned e.g. in~\cite{AvrKoh1992PRA, SchSchNagIto1995PRE})
\begin{equation}
\label{eq:kolmogorov}
	\mathcal{P}_{L,N}^{n}(s_1,s_2,\dots,s_n)=\sum_{\tau}\mathcal{P}_{L,N}^{n+1}(\tau,s_1,\dots,s_n)=\sum_{\tau}\mathcal{P}_{L,N}^{n+1}(s_1,\dots,s_n,\tau)
\end{equation}
for all $(s_1,\dots,s_n)\in\{0,1\}^n$.

Hence we get the distance-headway distribution
\begin{equation}
\label{eq:Awpphi}
	{\wp(d;\varrho)=\tfrac{\sigma}{\varrho}\big[1-\phi(0)\big]^{2}\phi(0)^{d-1}\,,\quad d\geq1\,,\quad
		\wp(0;\varrho)= 1-\tfrac{\sigma}{\varrho}[1-\phi(0)]\,.}
\end{equation}
Another needed quantity is the {Minimal-gap-length} probability
\begin{equation}
	P(n;\varrho)=\mathcal{P}_{\varrho}(1\overbrace{0\dots\underline{0}}^{n}\mid 0)=\mathcal{P}_{\varrho}^{n+1}(1\overbrace{0\dots0}^{n})/\sigma\,,\quad n\geq1\,,
\end{equation}
which is calculated as
\begin{equation}
	P(n;\varrho)=\frac{\varrho}{\sigma}\sum_{d=n}\wp(d;\varrho)=\frac{z}{\sigma}\left(1-\frac{z}{\sigma}\right)\,.
\end{equation}

\section{Time-headway derivation for genTASEP}
\label{app:thd}

Let us consider the forward variant of genTASEP. Under this update, the LP can hop more sites forward during one time step, it is possible that it hops from the site 1 in the same step as it hopped for the site 0, i.e., $k_\mathrm{out}^\mathrm{LP}(0)=k_\mathrm{out}^\mathrm{LP}(1)$. This happens with the probability
\begin{equation}
	\omega:=p\gamma\mathcal{P}_{\varrho,\to}(\underline{0}0|0)=p\gamma\mathcal{P}_{\sigma,\leftarrow}(\underline{1}1|1)=p\gamma\wp_{\leftarrow}(0;\sigma)=p\gamma(1-z/\sigma)\,.
\end{equation}
It is then convenient to calculate the time-headway probability as
\begin{equation}
	f(k)=\omega f(k\mid\tau_1=0)+(1-\omega)f(k\mid\tau_1=1)
\end{equation}

Let us first deal with the probability $f(k\mid\tau_1=0)$. Knowing that the site 1 is already empty, we know that the FP sitting in site $-n$ has to hop $(n+1)$-times forward during exactly $k$ steps. Let $C(n,k)$ stand for the conditional probability, that a particle performs $n$ forward hops during exactly $k$ steps, given that there are at least $n$ empty sites ahead. Then it reads
\begin{equation}
	f(k\mid\tau_1=0)=\sum_{n=1}^{+\infty}C(n+1,k)P(n)=\sum_{n=1}^{+\infty}(1-\phi)\phi^{n-1}C(n+1,k)\,,
\end{equation}
where we denote by $\phi:=1-z/\sigma$ for readability reasons.
\begin{lemma}
	For $n,k\geq1$ it follows
	\begin{equation}
		C(n,k)=p(p\gamma)^{n-1}q^{k-1}\sum_{a=0}^{\min\{n-1,k-1\}}{n-1 \choose a}{k-1 \choose a}\left(\frac{1-p\gamma}{q\gamma}\right)^a
	\end{equation}
\end{lemma}
\begin{proof}
	The considered particle needs to perform $k-1$ stops (denote as $\times$) and $n$ jumps (denote as $\circ$). Furthermore,  the $n$-th jump has to be performed after all $k-1$ stops. That means, there are $n-1+k-1$ that can be arbitrarily filled by $n-1$ symbols $\circ$ and $k-1$ symbols $\times$. Further, each pair $(\times\circ)$ in the sequence means that there is a hop performed with probability $p$. Similarly $(\circ\circ)$ means a hop with prob. $p\gamma$, $(\circ\times)$ a stop with prob. $1-p\gamma$, $(\times\times)$ a stop with prob. $1-p$. Let $n_{ab}$, $a,b\in\{\circ,\times\}$ denote the number of pairs $(ab)$ in the sequence. Fixing the number of pairs $n_{\times\circ}$, we get $n_{\circ\times}=n_{\times\circ}-1$, $n_{\circ\circ}=n-n_{\times\circ}$, and $n_{\times\times}=k-n_{\times\circ}$. Then
	\begin{equation}
		C(n,k)=\sum_{n_{\times\circ}=1}^{\min\{n,k\}}D(n,k,n_{\times\circ})p^{n_{\times\circ}}(p\gamma)^{n-n_{\times\circ}}q^{k-n_{\times\circ}}(1-p\gamma)^{n_{\times\circ}-1}\,,
	\end{equation}
	where $D(n,k,n_{\times\circ})$ stands for the number of sequences with $n_{\times\circ}$ pairs $({\times\circ})$. There are $n_{\times\circ}$ blocks of  $n_{\times\times}$ symbols $(\times\times)$ and $n_{\times\circ}$ blocks of $n_{\circ\circ}$ pairs $(\circ\circ)$. The number of all possible sequences is then
	\begin{equation}
		D(n,k,n_{\times\circ})={n_{\times\circ}+n_{\times\times}-1\choose n_{\times\circ}}{n_{\times\circ}+n_{\circ\circ}-1\choose n_{\times\circ}}\,.
	\end{equation}
\end{proof}

Then
\begin{eqnarray}
	f(k\mid\tau_1=0)&=&\sum_{n=1}^{+\infty}(1-\phi)\phi^{n-1}p(p\gamma)^{n}q^{k-1}\sum_{a=0}^{\min\{n,k-1\}}{n-1 \choose a}{k-1 \choose a}\left(\frac{1-p\gamma}{q\gamma}\right)^a\\
	&=& (1-\phi)p^2\gamma q^{k-1}\sum_{a=1}^{k-1}{k-1 \choose a}\left(\frac{1-p\gamma}{q\gamma}\right)^a\sum_{n=1}^{+\infty}{n\choose a}(p\gamma\phi)^{n-1}\,,
\end{eqnarray}
where we use the convention ${n\choose a}=0$ for $a>n$.
\begin{lemma}
	For $a\geq1$ it reads
	\begin{equation}
		\sum_{n=1}^{+\infty}{n\choose a}x^{n-1}=\frac{x^{a-1}}{(1-x)^{a+1}}\,.
	\end{equation}
\end{lemma}

Using preceding lemma we obtain
\begin{equation}
	f(k\mid\tau_1=0)=\frac{p(1-\phi)}{\phi(1-p\gamma\phi)}B_1^{k-1}-\frac{p(1-\phi)}{\phi}q^{k-1}\,,
\end{equation}
where
\begin{equation}
	B_1=\left(1-\frac{p(1-\phi)}{1-p\gamma\phi}\right)\,.
\end{equation}

For the calculation of $f(k\mid\tau_1=1)$ it is convenient to decompose the time interval $\Delta k$ as
\begin{equation}
	\Delta k = k_\mathrm{out}^\mathrm{FP}(0)-k_\mathrm{out}^\mathrm{LP}(0) =\Delta k_{1}+\Delta k_{2}\,,
\end{equation}
where
\begin{equation}
	\Delta k_{1} = k_\mathrm{in}^\mathrm{FP}(0)-k_\mathrm{out}^\mathrm{LP}(0)\\
\end{equation}
is the number of steps it takes the FP to enter the site 0 after LP left it at $k_\mathrm{out}^\mathrm{LP}(0)$ and 
\begin{equation}
	\Delta k_{2} = k_\mathrm{out}^\mathrm{FP}(0)-k_\mathrm{in}^\mathrm{FP}(0)	
\end{equation}	
is the number of steps it takes the FP to leave the site 0 after entering it at $k_\mathrm{in}^\mathrm{FP}(0)$. The time line is schematically depicted in Figure~\ref{fig:TimeLine}.

\begin{figure}[htb]
	\centering
		\includegraphics[scale=1]{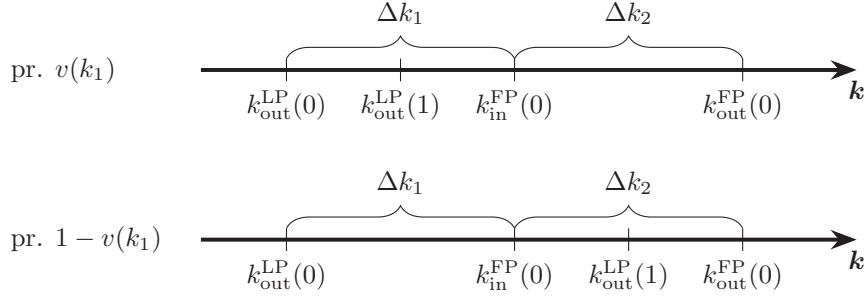}		
\caption{Mutual positions of events at $k_\mathrm{in/out}^\mathrm{LP/FP}(0/1)$ on the time line. Upper line with probability $v(k_1)$, lower line with probability $1-v(k_1)$.}
\label{fig:TimeLine}
\end{figure}

Following the technique used in~\cite{ChoPasSin1998EPJB} the probability $f(k\mid\tau_1=1;\varrho)$ can be calculated by means of the chain rule
\begin{equation}
\label{eq:fkgen}
	f(k\mid\tau_1=1;\varrho)=\sum_{k_{1}=1}^{k}g(k_1;\varrho)h(k-k_1\mid k_1;\varrho)\,,
\end{equation}
where
\begin{equation}
	g(k_1;\varrho)=\Pr(\Delta k_1=k_1\mid\varrho)\,,\qquad h(k_2\mid k_1;\varrho)=\Pr(\Delta k_2=k_2\mid \Delta k_1 =k_1;\varrho)\,.
\end{equation}

Using similar consideration as in the case $f(k\mid\tau_1=0)$ we obtain
\begin{eqnarray}
	g(k_1;\varrho)=\sum_{n=1}^{\infty}P(n;\varrho) C(n,k_1)=\frac{p(1-\phi)}{1-p\gamma\phi}B_1^{k_1-1}\,.
\end{eqnarray}

Let further $\psi:=z/\varrho$. Then
\begin{equation}
	b=p\mathcal{P}_{\varrho,\to}(\underline{1}0|1)=p\mathcal{P}_{\sigma,\leftarrow}(\underline{0}1|0)=pz/\varrho=p\psi
\end{equation}
stands for the probability that a particle hops to the neighbouring site within next step. Let us denote $B_2=(1-p\psi)$. The probability $v(k_1)$ that the site 1 will be emptied by LP before $k_1$-th step is $v(k_1)=1-B_2^{k_1-1}$. Then
\begin{eqnarray}
	h(0\mid k_1)&=&v(k_1)p\gamma=p\gamma(1-B_2^{k_1-1})\\
	h(k_2\mid k_1) &=& v(k_1)(1-p\gamma)C(1,k_2)+[1-v(k_1)]\sum_{\ell=0}^{k_2-1}b(1-b)^{\ell}C(1,k_2-\ell)\\
	&=&(1-B_2^{k_1-1})pq^{k_2}\frac{1-p\gamma}{1-p}+B_2^{k_1-1}\frac{p\psi}{1-\psi}(B_2^{k_2}-q^{k_2})\,.
\end{eqnarray}
Hence
\begin{eqnarray}
	f(k\mid 1)&=&p\gamma A_1B_1^{k-1}-p\gamma A_1(B_1B_2)^{k-1}\\
	&+&\sum_{k_1=1}^{k-1}\left[\frac{p}{q}(1-p\gamma)A_1 q^{k-1}\left(\frac{B_1}{q}\right)^{k_1-1}+A_1A_2B_2^{k-1}B_1^{k_1-1}\right.\\
	&-&\left.\frac{p}{q}(1-p\gamma)A_1 q^{k-1}\left(\frac{B_1B_2}{q}\right)^{k_1-1}-A_1A_2q^{k-1}\left(\frac{B_1B_2}{q}\right)^{k_1-1}\right]\,,
\end{eqnarray}
where
\begin{equation}
	A_1=\frac{p(1-\phi)}{1-p\gamma\phi}\,,\qquad A_2=\frac{p\psi}{1-\psi}\,.
\end{equation}

Here we notice that from the equation~\eqref{eq:z} it follows
\begin{equation}
	B_1B_2=\left(1-\frac{pz/\sigma}{1-p\gamma(1-pz/\sigma)}\right)\left(1-\frac{pz}{\varrho}\right)=q\,.
\end{equation}
Then obtaining the result~\eqref{eq:fkfinfor} is a question of technical computation.

\section*{References}
\bibliographystyle{elsarticle-num}
\bibliography{references}

\end{document}